\documentclass[reqno]{amsart}

\makeatletter
\renewcommand\@biblabel[1]{#1.}
\makeatother

\newtheorem{theorem}{Theorem}
\newtheorem{proposition}{Proposition}
\newtheorem{definition}{Definition}
\newcommand{\abs}[1]{\lvert#1\rvert}

\usepackage{hyperref}

\hypersetup{
    colorlinks=true,
    linkcolor=blue,
    filecolor=magenta,      
    urlcolor=cyan,
}
 
\begin{document}

\allowdisplaybreaks

\thispagestyle{plain}

\title[Holomorphic Path Integrals in Tangent Space for Flat Manifolds]{Holomorphic Path Integrals in Tangent Space for Flat Manifolds}
\author[G. Capobianco and W. Reartes]{Guillermo Capobianco \and Walter Reartes}
\address{Departamento de Matem\'atica, INMABB CONICET, Universidad Nacional del Sur, Av. Alem 1253, 8000 Bah\'ia Blanca, Buenos Aires, Argentina}
\email{\href{mailto:guillermo.capobianco@gmail.com}{guillermo.capobianco@gmail.com},\href{mailto:walter.reartes@gmail.com}{walter.reartes@gmail.com}}
\date{\today}


\begin{abstract}
In this paper we study the quantum evolution in a f\/lat Riemannian manifold. The holomorphic functions are def\/ined on the cotangent bundle of this manifold. We construct Hilbert spaces of holomorphic functions in which the scalar product is def\/ined using the exponential map.
The quantum evolution is proposed by means of an inf\/initesimal propagator and the holomorphic Feynman integral is developed via the exponential map. The integration corresponding to each step of the Feynman integral is performed in the tangent space. 
Moreover, in the case of $S^1$, the method proposed in this paper naturally takes into account paths that must be included in the development of the corresponding Feynman integral.
\end{abstract}

\maketitle

\tableofcontents

\section{Introduction}

The quantization of Riemannian manifolds is a very important topic in quantum mechanics in non-trivial spaces \cite{kleinert04,kobayashi69-2,kowalski-rembielinski00,kowalski-et-al98,mostafazadeh96}. From a geometrical point of view, the quantization in the phase space is natural, mainly because of its relation to the geometric quantization. The geometric quantization is strongly linked to the symplectic structure of the phase space \cite{woodhouse}.

Of particular interest is the case of cotangent bundles of Riemannian manifolds. The cotangent bundle carries a natural Riemannian structure which is the lifting of the structure at the base. It also has a compatible almost-complex structure $J$. This almost-complex structure gives rise to a complex structure if and only if the manifold is f\/lat \cite{gorbunov-et-al05}. This is the case we deal with in this paper.

We propose a quantization method for a f\/lat Riemannian manifold using the complex structure of the cotangent bundle to def\/ine Hilbert spaces of holomorphic functions def\/ined in the cotangent. To def\/ine the scalar product in this Hilbert space we make use of the exponential map, which allows us to perform an integration in the tangent space of the cotangent. The procedure is as follows. We assign to each point of the cotangent a Hilbert space by def\/ining a scalar product by integration in the tangent space using a Gaussian measure. These spaces are naturally isomorphic. With this product, each f\/iber is a Segal-Bargmann space for locally def\/ined holomorphic functions. The space of globally def\/ined holomorphic functions is a subspace of the former one. The space we get is a reproducing kernel Hilbert space (RKHS), like that of Segal-Bargmann. The existence of a reproducing kernel allows us to study the temporal evolution of the quantum wave function by means of Feynman integrals.
Moreover, we develop a method that takes into account all the paths that appear in the Feynman integral. 

In \cite{schulman}, the author studies how to propose a Feynman integral in non-simply connected conf\/iguration spaces that consider all the paths. For example, in the case of $S^1$ \cite{schulman}, the Green function obtained by studying a free particle, whose Lagrangian is $L=\frac{1}{2}\dot{\theta}^2$, is given by the following expression
\begin{equation}
G(\theta,T,\theta_0)=\frac{1}{\sqrt{2\pi iT}}e^{i\frac{(\theta-\theta_0)^2}{2T}}.
\end{equation}
However, it does not turn out to be the Green function for $S^1$, 
considering that the correct Green function must take into account all the paths, that is to say
\begin{equation}
G(\theta,T,\theta_0)=\sum_{\mathbb{Z}}G_n(\theta,T,\theta_0),
\end{equation}   
where
\begin{equation}
G_n(\theta,T,\theta_0)=\frac{1}{\sqrt{2\pi i T}}e^{i\frac{(\theta-\theta_0+2\pi n)^2}{2T}},
\end{equation}
Then, a Green function is f\/inally obtained which coincides with that resulting from the inf\/initesimal propagator proposed in \cite{reartes04}. In the construction of this propagator, the integration is performed in the tangent space of $S^1.$

In this paper we do the integration in the tangent space of the cotangent and therefore, for the aforementioned cases, all the paths are automatically considered.
\section{Quantization of the phase space}

\subsection{Metric and complex structure}\label{metricalevantada} 

In the case of a Riemannian manifold $Q$, as a conf\/iguration space, there is a metric in the cotangent space that is the natural lifting of the metric in $Q$. Also an almost-complex 
structure is obtained, which is compatible with the symplectic structure of the cotangent bundle. This structure is integrable when the curvature 
of  $Q$ is null \cite{gorbunov-et-al05, kobayashi69-2}.

The metric in $Q$ can be naturally lifted to the cotangent space $P$ as follows. (See
the details in the Appendix of \cite{capobiancoreartes2}).

We denote by $\sigma$ the metric on $Q$, and by $\sigma^\sharp$ the isomorphism induced by $\sigma$ between $P=T^*Q$ and $TQ$.
Then, we have a metric $G$ on $P$ given by
\begin{equation}\label{metricG}
G_m(V,W)=\sigma_q\left(T\pi V,T\pi W\right)+\sigma_q\left(\sigma^\sharp_q\frac{Dp_1}{\mathrm{d} t}(0),\sigma^\sharp_q\frac{Dp_2}{\mathrm{d} t}(0)\right),
\end{equation}
where $T\pi$ is the tangent application of the projection $\pi\colon P\to Q$, and $D$ is the covariant derivative.

A third structure appears naturally, the almost-complex structure~$J$. The triple $(J,G,\omega)$ is known as a compatible triple.  
That is, given the f\/ields~$V$ and~$W$, they verify $G(V,W)=\omega(V,JW)$.

Complexifying the tangent space $T_mP$, we have the space $T_mP \otimes \mathbb{C}$. 
The map $J_m$ can be extended naturally, and we have the decomposition
$T_mP\otimes \mathbb{C}= T_m^{(1,0)}P\oplus T_m^{(0,1)}P$ for all $m \in P$ which can be extended to the whole f\/iber bundle $TP$. Then,
the complexif\/ied tangent bundle, $T^\mathbb{C}P$, splits as follows (see \cite{kobayashi69-2})
\begin{gather*}
T^\mathbb{C}P=T^{(1,0)}P\oplus T^{(0,1)}P,
\end{gather*}
where $T^{(1,0)}P$ and $T^{(0,1)}P$ are the images of the projections $\Pi^+$ and $\Pi^-$ given by
$\Pi^\pm = \frac{1 \mp \mathrm{i} J}{2}$.

The projection $\Pi^+$ is a natural isomorphism between $T_mP$ and 
$T^{(1,0)}_mP$ ({\it holomorphic tangent space}). That is, given a vector $V=(\dot{q}^1,\ldots,\dot{q}^n,\dot{p}_1,\ldots,\dot{p}_n)\in T_mP$,
$\Pi^+V=\dot{z}^i\frac{\partial}{\partial z^i},$
where the induced complex coordinates are
\begin{gather}\label{coordenadascomplejas}
\dot{z}^i=\dot{q}^i+\mathrm{i}\sigma^{im}(\dot{p}_m-p_k\Gamma_{ml}^k\dot{q}^l),
\end{gather}
and the corresponding holomorphic vector f\/ields are
\begin{gather*}
\frac{\partial}{\partial z^i} = \frac{1}{2}\left(\frac{\partial}{\partial q^i} + p_k\Gamma^k_{ij}\frac{\partial}{\partial p^j} - \mathrm{i}\sigma_{ij}\frac{\partial}{\partial p^j}\right),
\end{gather*}
where $\sigma_{ij}$ and $\Gamma^k_{ij}$ are the matrix coef\/f\/icients of the metric and the Christof\/fel symbols respectively (the Einstein summation convention is used).
Computing the Lie bracket of the above-mentioned f\/ields, we obtain
\begin{gather*}
\left[\frac{\partial}{\partial z^i},\frac{\partial}{\partial z^j}\right] = \mathrm{i} R^m_{kij}p_m\sigma^{lk}\left(\frac{\partial}{\partial z^l}-\frac{\partial}{\partial\bar{z}^l}\right).
\end{gather*}
Then, by the Nirenberg-Newlander theorem \cite{kobayashi69-2}, the distribution is integrable if and only if the curvature tensor of the metric~$\sigma$ is identically null, see \cite{gorbunov-et-al05}.
That is, the base manifold is f\/lat.  
In this paper we assume that the conf\/iguration space is a f\/lat manifold. 

Recently, the research on f\/lat Riemannian manifolds (known as {\em euclidean space forms} \cite{wolf06,kuhnel06}) has shown many developments. 
For instance, its spectral properties are studied in detail in \cite{miatellopodesta06,miatellorossetti09}.
In cosmology, the {\em euclidean space forms} are used to model the spatial part of f\/lat universe models \cite{ellis71,levin02,levin98,levin98_2,levin98_3}.
And more recently a quantization procedure for {\em euclidean space forms} based on path integrals is developed in \cite{capobiancoreartes2}.

\subsection{Flat Riemannian manifolds}

The following theorem by W.~Killing and H.~Hopf~ characterizes f\/lat Riemmanian manifolds (it is a part of a~more general theorem \cite{wolf06}).

\begin{theorem}
Let $M$ be a Riemannian manifold of dimension $n\geq2$ and zero curvature. Then~$M$ is complete and connected if and only if it is isometric to the quotient $\mathbb{R}^n/\Gamma$ with $\Gamma \subset E(n)$, where~$\Gamma$ acts freely and properly discontinuously.
\end{theorem}

These manifolds are known as Euclidean space forms and $E(n)$ is the Euclidean group (semidirect product of the groups~$O(n)$ and $\mathbb{R}^n$). 

In one dimension these manifolds are the real line $\mathbb{R}$ and the circle $S^1$ while in dimension~$2$ there are f\/ive manifolds, the plane~$\mathbb{R}^2$, the cylinder, the inf\/inite M\"oebius strip, the torus and the Klein bottle. In dimension $3$ there are $18$ types,  $10$ of which are compact, $6$~orientable, and $4$ non-orientable~\cite{kuhnel06,wolf06}. In higher dimensions the number grows signif\/icantly. 

Every homogeneous Riemannian manifold is dif\/feomorphic to some Lie group but in general a space form is not necessarily homogeneous. 
In particular, when the space form is homogeneous of dimension~$n$, then it is isometric to the product $\mathbb{R}^m\times T^{n-m}$ of a Euclidean space with a f\/lat 
Riemannian torus \cite[p.~88]{wolf06}.

This paper focuses on orientable compact f\/lat manifolds. 
Compact f\/lat Riemannian manifolds of dimension $n$ are quotients of polyhedra in~$\mathbb{R}^n$ by identifying faces (see \cite[p.~99]{wolf06}). 
Functions def\/ined on the manifold are functions on~$\mathbb{R}^n$, which are invariant under the action of the group.

An important invariant for a compact Euclidean space form is its volume. This can be def\/ined in terms of a fundamental region for~$\Gamma$ 
in~$\mathbb{R}^n$~\cite{mcmullen02}.
The volume of a space form $\mathbb{R}^n/\Gamma$ is def\/ined to be the volume of any
fundamental region,
$c_\gamma$, 
\begin{gather*}
c_\gamma:=\big\{x\in \mathbb{R}^n; \|\gamma(0)-x\|\leq\|\gamma'(0)-x\|\,\text{for every}\, \gamma' \in \Gamma\big\},
\end{gather*}
where $\gamma (0)$ is the action of $\gamma$ on $0\in\mathbb{R}^n$. $c_\gamma$ is an $n$-dimensional convex polyhedron in $\mathbb{R}^n$ bounded by hyperplanes which are perpendicular bisectors of line segments $[\gamma(0), \gamma'(0)]$. Its bounda\-ry~$\partial c_\gamma$ carries a locally f\/inite decomposition into convex polyhedra of dimension $n-1$. The space form~$\mathbb{R}^n/\Gamma$ is then obtained from~$c_\gamma$ by identifying points in~$\partial c_\gamma$ which are equivalent modulo~$\Gamma$.

The family $\{c_\gamma\}$ forms a crystalline structure (see \cite[p.~100]{wolf06}), 
that can be generated by translation of a f\/inite set of vectors def\/ining the crystal lattice. This set forms a basis of~$\mathbb{R}^n$. 
Dual basis vectors multiplied by $2\pi$ are the basis of the reciprocal lattice,~$\mathcal{L}$. Let~$K$ be an element of the reciprocal lattice,
then a function with the symmetry of this lattice has a Fourier expansion given by
\begin{gather}\label{f}
f(x)=\sum_{K\in\mathcal{L}} c_K \mathrm{e}^{\mathrm{i}  K\cdot x}.
\end{gather}
This function is well def\/ined on the manifold if it is also invariant under the action of $\Gamma$, i.e., $(\gamma f)(x)=f(\gamma x)=f(x)$ for all $\gamma\in\Gamma$.

The $3$-dimensional orientable compact Euclidean space forms present a~parti\-cu\-lar interest for cosmology, 
since they could model the spatial part of the f\/lat-universe models~\cite{ellis71}.
Recently J.~Levin et al. seek to develop 
a plausible cosmological model using orientable compact Euclidean space forms of 
dimension $3$ in agreement with results of observations made on the cosmic 
microwave background radiation~\cite{levin02, levin98,levin98_2,levin98_3}.

\subsection{The Hilbert space}
The cotangent bundle $P$ with the metric $G$ is a complete Riemannian manifold. 
Then, by the theorem of Hopf-Rinow, it is geodesically complete, that is, for all $m\in P$ there exists a well def\/ined exponential map on the entire tangent space $T_mP.$

Here, the tangent vectors are represented by its complex coordinates (\ref{coordenadascomplejas}) using the projection $\Pi^+$.
The square modulus $|z|^2$ is given by
\begin{equation}
|z|^2=G_m(z,z)=\sigma_{ij}(\pi m)z^i\bar{z}^j,
\end{equation}
where $\sigma_{ij}$ are the components of the matrix corresponding to the metric in the base manifold $Q$.

The natural volume in the vectorial space $T_mP$ is given by the {\em pull-back} of the Riemannian metric by the exponential map (which coincides with the Liouville volume). 

\begin{proposition}
The pullback of the Riemannian volume by the exponential map in $m$ is given by
\begin{equation}
\exp_m^*\Omega(z) = \sigma(\pi m)\,dz.
\end{equation}
where $\sigma$ is the determinant of the metric, and $dz$ is the Lebesgue measure in $\mathbb{C}^n.$
\end{proposition}

We use the exponential map in order to def\/ine a product of holomorphic functions as follows. 

\begin{definition}
 Let $m \in P$, then for holomorphic functions $\phi$ and $\psi$ def\/ined 
on $P$ we def\/ine the following normalized scalar product 
\begin{equation}\label{productoescalar2}
\left<\phi,\psi\right> = \int_{T_mP}\overline{\phi(\exp_m z)}\,\psi(\exp_m z)\,e^{-|z|^2}\,dz.
\end{equation}
\end{definition}

Even though the def\/inition depends on the point $m$, the dif\/ferent products are isometric to each other.
In the following theorem we demonstrate that all these elements constitute a Hilbert space of holomorphic functions.
\begin{theorem}\label{hilbert} 
Let $Q$ be a f\/lat Riemannian manifold, connected and geodesically complete. We assume the surjectivity of the exponential map.
Then, the space of square integrable holomorphic functions on $P$ is a Hilbert space with the norm associated to the scalar product (\ref{productoescalar2}). We call this Hilbert space $\mathcal{B}_{P}.$
\end{theorem}
\begin{proof}
A holomorphic function $\phi : P \rightarrow \mathbb C$ induces (by {\em pull-back}) a holomorphic function on $T_mP,$ that is $\tilde{\phi}= \phi \circ \exp_m $. The set of all holomorphic functions on $T_mP$ whit the scalar product (\ref{productoescalar2}) is a Hilbert space $\mathcal{B}=\mathcal{H}L^2(\mathbb{C}^n,e^{-|z|^2})\subset L^2(\mathbb{C}^n,e^{-|z|^2})$ (see \cite{hall00a}). Let $\{\phi_i\}_{i=1}^\infty$ be a Cauchy sequence of holomorphic functions on $P$. It induces a new sequence $\{\tilde\phi_i\}_{i=1}^\infty$ which converges to a holomorphic function $\tilde{\phi}$ on $\mathcal{B}.$ One of the key properties of these spaces of holomorphic functions is the continuity of the evaluation map, i.\ e., for all $z$ exists a constant $M_z$ such that for all $\tilde{\phi}$ it is verif\/ied that
\begin{equation}
|\tilde{\phi}(z)|^2\leq M_z \|\tilde{\phi}\|_{L^2}^2.
\end{equation}
Then, 
\begin{equation}\label{anterior2}
|\tilde{\phi}_i(z)-\tilde{\phi}(z)|^2\leq M_z \|\tilde{\phi}_i - \tilde{\phi}\|_{L^2}^2 \rightarrow 0.
\end{equation}
Then for all $n\in P$ we def\/ine $\phi(n)=\tilde{\phi}(\exp_m^{-1}n).$ The function $\phi$ is well def\/ined. Indeed, if $z_1, z_2\in \exp_m^{-1}n$, we have $\tilde{\phi}_i(z_1)=\tilde{\phi}_i(z_2),$ 
then by (\ref{anterior2}), $\tilde\phi(z_1)=\tilde\phi(z_2).$ 
\end{proof}

We can consider $\mathcal{B}_{P}$ as a subspace of $\mathcal{B}$ by means of the exponential map. In $\mathcal{B}$ exists a function called reproducing kernel $K(z,w)$ (see \cite{hall00a}) with the following property 
\begin{equation}\label{rkp}
\tilde{\phi}(z)=\int_{\mathbb{C}^n}K(z,w)\,\tilde{\phi}(w)\,e^{-|w|^2}dw
\end{equation}
for all $\tilde{\phi} \in \mathcal{B}$.
Also $K(w,z)= \overline{K(z,w)}$ and satisf\/ies the following composition rule

 \begin{equation}
K(z,u)=\int_{\mathbb{C}^n}K(z,w)\,K(w,u)\,e^{-|w|^2}dw
\end{equation}
for all $z,u \in \mathbb{C}^n $. 
Also for all $z \in \mathbb{C}^n$, $\;|\tilde{\phi}(z)|^2 \leq K(z,z)\; \|\tilde{\phi}\|^2$ 
($K(z,z)$ is optimal, that is, for each $z$, there exists a non-zero $\tilde{\phi}_z \in \mathcal{B}$ for which equality holds).
Moreover, the reproducing kernel acts as a projector, i.\ e., if $\tilde{\phi} \in L^2(\mathbb{C}^n,e^{-|z|^2})$,
and calling $\mathcal{P}\tilde{\phi}$ to the orthogonal projection of $ \tilde{\phi}$ onto $\mathcal{B}$, then
\begin{equation}
\mathcal{P}\tilde{\phi}(z)=\int_{\mathbb{C}^n}K(z,w)\,\tilde{\phi}(w)\,e^{-|w|^2}dw.
\end{equation}

The reproducing kernel $K(z,w)$ is unique in the following sense. Given any $z \in \mathbb{C}^n$, if $F_z(\;\cdot\;) \in \mathcal{B} $ satisf\/ies
\begin{equation}
\tilde{\phi}(z)=\int_{\mathbb{C}^n}\overline{F_z(w)}\,\tilde{\phi}(w)\,e^{-|w|^2}dw
\end{equation}
for all $\tilde{\phi} \in \mathcal{B}$, then $\overline{F_z(w)}=K(z,w)$.

But in general the reproducing kernel $K(z,w)$ is not the pull-back of a function on $P$.

A reproducing kernel in the space $\mathcal{H}_P$, is a function
\begin{equation}
K_{P}\colon P\times P\to \mathbb{C},
\end{equation}
which is holomorphic in the f\/irst coordinate and antiholomorphic in the second, and such that the following equation holds for all $\phi\in\mathcal{B}_{P}$
\begin{equation}
\phi(n)=\int_{T_mP}K_{P}(n,\exp_m w)\,\phi(\exp_m w)\,e^{-|w|^2}dw.
\end{equation}
The following theorem allows us to obtain the reproducing kernel 

\begin{theorem} Let $\{e_j\}_{j=1}^\infty$ be an orthonormal basis for $\mathcal{B}_{P}$, then for $m,$ $n \in P$ 
\begin{equation}
 \sum_{j=1}^\infty |e_j(m)\overline{e_j(n)}| < \infty
\end{equation}
and
\begin{equation}
 K_{P}(m,n)=\sum_{j=1}^\infty e_j(m)\overline{e_j(n)}.
\end{equation}
\end{theorem}

\begin{proof}
Given that $\mathcal{B}_P$ is a Hilbert subspace of $\mathcal{B}$, we have the following direct sum 
\begin{equation}
\mathcal{B} = \mathcal{B}_P \oplus \mathcal{B}_P^{\bot}.
\end{equation}
Let $\{f_j\}_{j=1}^\infty$ be an orthonormal basis of $\mathcal{B}$ compatible with this decomposition. Then, there exists a reproducing kernel $K(z,w)$ in this space (see \cite{hall00a}) given by
\begin{equation}
  K(z,w)=\sum_{j=1}^\infty f_j(z)\overline{f_j(w)}.
\end{equation}
Then we have 
\begin{equation}
  K(z,w)=\tilde{K}_P(z,w) + \tilde{K}_P^\bot(z,w).
\end{equation}
Considering that $\tilde{K}_P^\bot(z,w)$ acting on a function $\phi \in \mathcal{B}_P$ does not contribute to the integral, we have the following result
\begin{equation}
\begin{aligned}
\phi(n) &=\int_{T_mP}K(\exp_m^{-1}n,w)\phi(\exp_m w)e^{-\abs{w}^2}\,dw \\
&= \int_{T_mP}\tilde{K}_P(\exp_m^{-1}n,w)\phi(\exp_m w)e^{-\abs{w}^2}\,dw\\
\end{aligned}
\end{equation}
and $K_P(n,q)=\tilde{K}_P(\exp_m^{-1}n,\exp_m^{-1}q)$ is the appropriate reproducing kernel. 
\end{proof}

\subsection{The $S^1$ case} 

Consider the case of a particle whose conf\/iguration space is the unit circle. Then the phase space is a cylinder. The euclidean metric on the circle induces a euclidean 
metric on the cylinder. The system is interesting due to the non trivial topology, in particular it is a non-simply connected manifold. We choose the following coordinates
\begin{equation}
-\pi  < q < \pi  \quad\text{y}\quad -\infty < p < \infty.
\end{equation}
And the following coordinates on the tangent space of the cylinder at the point $(0,0)$, called $x$ and $y$ respectively. That is, we can write a vector as follows
\begin{equation}
V=x\frac{\partial}{\partial q}+y\frac{\partial}{\partial p}.
\end{equation}
The exponential map sends the point $(x,y)$ to the point $q=x \mod 2\pi$ and $p=y$ on the previous coordinate chart. In the tangent space of the cylinder we choose the following complex coordinates
\begin{equation}
\begin{aligned}
z = x - iy\\
\bar{z} = x + iy.
\end{aligned}
\end{equation}

Now we consider the space of holomorphic functions on the cylinder, lifted to the tangent space by the exponential map. They are holomorphic in $z$ and periodic in the real part, $x$ coordinate. Then the scalar product in the space of holomorphic functions on the cylinder is
\begin{equation}\label{lift}
\left<\phi,\psi\right> = \frac{1}{2\pi} \int_\mathbb{C}\overline{\phi(z)}\psi(z)e^{-z\bar{z}} dz.
\end{equation}

This space is similar to the Segal-Bargmann space, moreover, the functions are periodic in the real part. In view of theorem~\ref{hilbert}, it is a Hilbert subspace of $\mathcal{B}$. Actually, in the expression \eqref{lift}, the functions $\phi$ and $\psi$ in the integral are the functions $\phi$ and $\psi$ lifted to the tangent space by the exponential map respectively. For simplicity we call they the same way. With this scalar product the set of square integrable holomorphic functions on the cylinder is a Hilbert space. We call it $\mathcal{H}_P$.  

A function in $\mathcal{H}_P$ has the following form
\begin{equation}\label{ff}
\psi(z) = \sum_{n=-\infty}^\infty c_n e^{inz}.
\end{equation}

The functions 
\begin{equation}
\phi_n(z)=e^{inz} = e^{in(x-iy)},\quad n\in\mathbb{Z}
\end{equation}
are periodic in the coordinate $x$ with period $2\pi.$
The set $\{\phi_n\}$ is a non-orthogonal basis of $\mathcal{H}_P$, in fact
\begin{equation}
\left<\phi_p,\phi_q\right> = e^{pq}.
\end{equation}

Considering that $\|\phi_p\|=e^{p^2/2},$ then we have the normalized functions $\tilde{\phi}_p$
\begin{equation}
\tilde{\phi}_p(z) = \frac{\phi_p(z)}{\|\phi_p\|} = e^{ipz-p^2/2}.
\end{equation}
With these functions, the scalar product is
\begin{equation}
\left<\tilde{\phi}_p,\tilde{\phi}_q\right> = e^{-(p-q)^2/2}.
\end{equation}

\subsubsection{Reproducing kernel}

It can be directly verif\/ied that the reproducing kernel is
\begin{gather*}
K(z,\overline{w})=\frac{1}{2\pi}\int_{-\pi}^\pi\frac{\rho_t^z(x)\rho_t^{\overline{w}}(x)}{\rho_t^{x_0}(x)}\mathrm{d} x,
\end{gather*}
where
\begin{gather*}
\rho_t^z(x)=\frac{1}{2\pi}\sum_{k \in\mathbb{Z}}e^{ik(x-z)-k^2\frac{t}{2}}
\end{gather*}
is the solution of the heat equation on $S^1$ with $t=1$ \cite{grigoryan06,grigoryan09}.

That is, given a function $\phi(z)$ as in \eqref{ff}, it verif\/ies \eqref{rkp}, i.\ e.

\begin{equation}
 \phi(z)=\sum_{k\in \mathbb{Z}}\phi_ke^{ikz} = \int_{\mathbb{C}}K(z,\overline{w})\phi(w)e^{-|w|^2}dw
\end{equation}

Indeed,
\begin{equation}
\begin{split}
\phi(z)&=\frac{1}{(2\pi)^4}\int_{\mathbb{C}}\int_{-\pi}^{\pi}\frac{1}{\rho_1^{x_0}(x)}\sum_{k,l,m\in\mathbb{Z}}e^{-ik(x-z)-\frac{k^2}{2}}e^{il(x-\overline{w})-\frac{l^2}{2}}\phi_me^{imw-w\overline{w}}dw\\
&= \sum_{k\in \mathbb{Z}}\phi_ke^{ikz}
\end{split}
\end{equation}

It is interesting to note that the reproducing kernel $K(z,\overline{w})$ coincides with that found in \cite{capobiancoreartes2}, although in that paper the reproducing kernel is found using a generalized Segal-Bargmann transform \cite{hall00a,hall94,hall97a,hall00b,hallkirwin11,hall05}
and the scalar product in that case is def\/ined by an integration in the manifold, in which integrand the solution of the heat equation explicitly appears.

\subsubsection{Ladder Operators}

Now we can def\/ine the ladder operators $a^{+}$ and $a$ (creation operator and annihilation operator respectively).

Given an orthonormal basis $\{\beta_i\}$ of $\mathcal{H}_P$, we have the following expression for the reproducing kernel
\begin{equation}
K(z,\bar{w}) = \sum_{n=0}^{\infty} \beta_n(z)\overline{\beta_n(w)} = \overline{\zeta_{\bar{z}}(w)}. 
\end{equation}

And the operators are the following
\begin{equation}
a^{+} \psi(z) = \left<\zeta_{\bar{z}},w\psi\right>_{L^2} =\frac{1}{2\pi } \sum_{n=0}^{\infty} \beta_n(z)\int_{\mathbb{C}}\overline{\beta_n(w)} w \psi(w)e^{-w\bar{w}}dw,
\end{equation}
and
\begin{equation}
a \chi(z) = \left<\zeta_{\bar{z}},\frac{d\chi}{dw}\right>_{L^2} =\frac{1}{2\pi} \sum_{n=0}^{\infty} \beta_n(z)\int_{\mathbb{C}}\overline{\beta_n(w)} \frac{d\chi}{dw}(w)e^{-w\bar{w}}dw = \frac{d\chi}{dz}(z).
\end{equation}

The functions $\psi$ and $\chi$ are given by the following series
\begin{equation}
\psi(z) = \sum_{p=0}^{\infty} u_p\beta_p(z), \quad\text{}\quad \chi(z) = \sum_{q=0}^{\infty} v_q\beta_q(z).	
\end{equation}

Then
\begin{eqnarray}	
\left< a^{+} \psi,\chi \right> &=& \frac{1}{(2\pi)^2} \sum_{n=0}^{\infty} \int_{\mathbb{C}} \overline{w\psi(w)} \beta_n(w) e^{-w\bar{w}}dw \int_{\mathbb{C}} \overline{\beta_n(z)}\chi(z)e^{-z\bar{z}} dz \nonumber\\ 
&=& \frac{1}{(2\pi)^2} \sum_{n=0}^{\infty} \int_{\mathbb{C}} \overline{\psi(w)} \frac{d\beta_n}{dw}(w) e^{-w\bar{w}}dw\int_{\mathbb{C}} \overline{\beta_n(z)}\chi(z)e^{-z\bar{z}} dz, 
\end{eqnarray}
where the last expression was obtained integrating by parts.
Then, using the series expansion of the functions $\psi$ and $\chi$, we obtain
\begin{eqnarray}
\left< a^{+} \psi,\chi \right> &=& \frac{1}{(2\pi)^2} \sum_{n,p,q=0}^{\infty}  \bar{u}_p v_q \int_{\mathbb{C}} \frac{d\beta_n}{dw}(w) \overline{\beta_p(w)} e^{-w\bar{w}} dw \int_{\mathbb{C}} \overline{\beta_n(z)} \beta_q(z) e^{-z\bar{z}} dz  \nonumber\\ 
&=& \frac{1}{2\pi} \sum_{p,q=0}^{\infty}  \bar{u}_p v_q \int_{\mathbb{C}} \frac{d\beta_q}{dw}(w) \overline{\beta_p(w)} e^{-w\bar{w}} dw,
\end{eqnarray}
where we use the orthonormality of $\beta_i$'s.

Also
\begin{equation}
\left< \psi,a\chi \right> =  
\frac{1}{(2\pi)^2} \sum_{n=0}^{\infty} \int_{\mathbb{C}} \overline{\beta_n(w)} \frac{d\chi}{dw}(w) e^{-w\bar{w}}dw\int_{\mathbb{C}} \overline{\psi(z)}\beta_n(z)e^{-z\bar{z}} dz. 
\end{equation}
where we use again the series expansion of $\psi$ and $\chi$,
\begin{eqnarray}
\left< \psi,a\chi \right> &=& \frac{1}{(2\pi)^2} \sum_{n,p,q=0}^{\infty}  \bar{u}_p v_q \int_{\mathbb{C}} \frac{d\beta_q}{dw}(w) 
\overline{\beta_n(w)} e^{-w\bar{w}} dw \int_{\mathbb{C}} \overline{\beta_p(z)} \beta_n(z) e^{-z\bar{z}} dz \nonumber\\
&=& \frac{1}{2\pi} \sum_{p,q=0}^{\infty}  \bar{u}_p v_q \int_{\mathbb{C}} \frac{d\beta_q}{dw}(w) \overline{\beta_p(w)} e^{-w\bar{w}} dw. 
\end{eqnarray}
Then we have the following result
\begin{equation}
\left< a^{+} \psi,\chi \right> = \left< \psi,a\chi \right>,
\end{equation}
i.\ e., $a^{+}$ and $a$ are adjoint operators.

\section{Path Integral}

Let $Q$ be a Riemannian manifold, connected and geodesically complete. Let $P$ be the cotangent bundle of $Q$.
In section (\ref{metricalevantada}) we shown that there is a natural lifting of the metric $\sigma$ in $Q$ to a metric $G$ in $P$ (see \eqref{metricG}).
Also, every simplectic manifold is almost-K\"ahler, then $P$ admits a compatible almost-complex structure $J$. 
We choose $J$ such that the metric $G=\omega(\cdot,J\cdot)$ coincides with the lifting of the metric in $Q$.  If $J$ is integrable, $P$ is a K\"ahler manifold.

Now we study the temporal evolution of the quantum wave function. 
Let $K\colon P\times P\to\mathbb{C}$ be the reproducing kernel, then  
\begin{equation}
\phi(m) = \int_{T_mP}K(m,\exp_m z)\phi(\exp_m z) e^{-|z|^2}\,dz.
\end{equation} 

Let $H$ be the hamiltonian operator in $\mathcal{H}L^2(P),$ represented by the integral kernel $K_H$, i.\ e. 
\begin{equation}
H\phi(m) = \int_{T_mP}K_H(m,\exp_m z)\phi(\exp_m z) e^{-|z|^2}\,dz.
\end{equation} 

According to the solution of the Schr\"odinger equation, the evolution operator associated to time-independent hamiltonians is
\begin{equation}
U_t=e^{-iHt}.
\end{equation}

For small time intervals we have the following approximation
\begin{equation}
U_\Delta \cong 1-iH\Delta.
\end{equation}
Then, we def\/ine the inf\/initesimal evolution operator as follows, \cite{capobiancoreartes2}
\begin{equation}
u_\Delta\phi(m) = \int_{T_mP}K(m,\exp_m z)\phi(\exp_m z) e^{-|z|^2}e^{-i\Delta K_H/K}\,dz.
\end{equation}

\begin{multline}\label{uu}
u_\Delta(u_\Delta\phi)(m) = \int_{T_mP}K(m,\exp_m z)\int_{T_{(\exp_mz)}P}K(\exp_m z,\exp_{(\exp_m z)}z_2) \\
\phi(\exp_{(\exp_m z)}z_2)
e^{-i\Delta K_H(\exp_m z,\exp_{(\exp_m z)}z_2)/K(\exp_m z,\exp_{(\exp_m z)}z_2)}\\
e^{-i\Delta K_H(m,\exp_m z)/K(m,\exp_m z)}e^{-|z_2|^2}\,dz_2 e^{-|z|^2}\,dz.
\end{multline}

Now, we use the following notation 

$$
\exp_m z=m_2
$$

$$
\exp_{(\exp_m z)}z_2=m_3
$$

$$
\exp_{(\exp_{(\exp_m z)}z_2)}z_3=m_4
$$

In that way we can def\/ine in general $m_j$. It is important to note that $m_j=m_j(z_{j},z_{j-1},...,z_2,z)$. 
This allows us to write (\ref{uu}) as follows

\begin{multline}
u_\Delta(u_\Delta\phi)(m) = \\
\int_{T_mP}\!\!\!\!K(m,m_2)\int_{T_{m_2}P}\!\!\!\!K(m_2,m_3) \phi(m_3)
e^{-i\left(\frac{\Delta K_H(m_2,m_3)}{K(m_2,m_3)}+\frac{\Delta K_H(m,m_2)}{K(m,m_2)}\right)} e^{-|z_2|^2}\,dz_2e^{-|z|^2}\,dz.
\end{multline}

Then, after $n$ compositions, we have

\begin{equation}
u_\Delta^n\phi(m) = 
\int_{T_mP}...\int_{T_{m_n}P} \phi(m_{n+1}) e^{-i\sum_{i=1}^{n}\frac{\Delta K_H(m_i,m_{i+1})}{K(m_i,m_{i+1})}}\prod_{i=1}^{n}K(m_i,m_{i+1})e^{-|z_i|^2}\,dz_i.
\end{equation}

Where in the last equation we identify $z$ as $z_1$ and $m$ as $m_1$ respectively.
The composition of these inf\/initesimal evolution operators allows us to obtain the evolution operator for f\/inite time $t$ as 
\begin{equation}
U_t\phi(m) = \lim_{n\to\infty}(u_{t/n})^n\phi(m).
\end{equation}

It is important to mention that the tangent space corresponding to a f\/lat manifold is also a f\/lat manifold, so in our evolution proposal we only deal with euclidean integrals 
in each step of the iteration.

Specif\/ically, the six $3$-dimensional orientable compact Euclidean space forms are the following quotient spaces $\mathbb{R}^3/\Gamma_i$, $i=1,\dots,6$ 
(see \cite[p.~117]{wolf06} and \cite[p.~302]{kuhnel06}). For example, for the torus $T^3$ (which is constructed by identifying the opposite 
faces of a parallelepiped by translations), $\Gamma_1$ is generated by three translations $t_1$, $t_2$, $t_3$, in the direction of three linear independent 
vectors. And $\Gamma_3$ is generated by $\Gamma_1$ and a screw motion $\alpha^3=t_3$. 

For instance, when the conf\/iguration space is the manifold $\mathbb{R}^3/\Gamma_3$, $P=T^*\left(\mathbb{R}^3/\Gamma_3\right)$ and

\begin{multline}\label{R3G3}
u_\Delta^n\phi(m) = \\
\int_{T_mT^*\left(\mathbb{R}^3/\Gamma_3\right)}\!\!\!\!\!\!...\int_{T_{m_n}T^*\left(\mathbb{R}^3/\Gamma_3\right)}\!\!\!\! \phi(m_{n+1}) e^{-i\sum_{i=1}^{n}\frac{\Delta K_H(m_i,m_{i+1})}{K(m_i,m_{i+1})}}\prod_{i=1}^{n}K(m_i,m_{i+1})e^{-|z_i|^2}\,dz_i,
\end{multline}
where $z_i \in \mathbb{C}^3$ and $\phi(m)$ is a function in $\mathbb{C}^3$ that preserves the symmetry of the euclidean space form and 
is invariant under the action of $\Gamma_3$.

\section{Discussion}

Recently, the research on f\/lat Riemannian manifolds has shown many developments \cite{miatellopodesta06,miatellorossetti09,ellis71,levin02,levin98,levin98_2,levin98_3}.
For instance, a quantization scheme for {\em euclidean space forms} based on path integrals is developed in \cite{capobiancoreartes2}.

In the present paper we make a dif\/ferent and original proposal to the quantization of f\/lat manifolds. It is interesting to note that it shares some results (such as the reproducing kernel) with the scheme studied in \cite{capobiancoreartes2}.

Here we study the holomorphic quantization of a quantum state whose configuration space is a f\/lat Riemannian manifold. The Hilbert spaces are obtained using the exponential map. To def\/ine the scalar product in this Hilbert space we make use of the exponential map, 
which allows us to perform an integration in the tangent space of the cotangent. Also the Feynman integral is developed via the exponential map. The existence of a reproducing kernel allows us to study the temporal evolution of the quantum wave function. As an advantage, the calculations are simpler. In \cite{schulman}, for example, the author studies the Feynman integral in non-simply connected conf\/iguration spaces and the obtained Green function must be corrected in order to include all the paths.
In the case of $S^1$, the method proposed in our paper naturally takes into account all the paths.

Finally, in the last section we apply our quantization method to the case when the conf\/iguration space is the manifold $\mathbb{R}^3/\Gamma_3$, and we show the evolution operator in (\ref{R3G3}) as a composition of inf\/initesimal evolution operators.

\subsection*{Acknowledgements}

This work was supported by the Universidad Nacional del Sur (Grants PGI 24/L096 and PGI 24/ZL10).

\end{document}